\def\FORCEFULL{}%
\pdfoutput=1
\newif\ifprlmode
\ifdefined\FORCEFULL
  \prlmodefalse   %
\else
  \prlmodetrue    %
\fi

\ifprlmode
  \documentclass[aps,prl,twocolumn,superscriptaddress]{revtex4-2}
\else
  \documentclass[aps,pra,onecolumn,superscriptaddress,longbibliography]{revtex4-2}
\fi

\usepackage{graphicx}
\usepackage{amsmath,amssymb}
\usepackage{amsthm}
\usepackage{placeins}
\usepackage{quantikz}
\usetikzlibrary{shapes.geometric, arrows}
\usepackage[colorlinks=true,linkcolor=blue,citecolor=blue]{hyperref}
\usepackage{cleveref}
\usepackage[ruled,vlined]{algorithm2e} %
\SetKwInput{KwInput}{Input}
\SetKwInput{KwOutput}{Output}

\newtheorem{theorem}{Theorem}
\newtheorem{lemma}[theorem]{Lemma}
\newtheorem{corollary}[theorem]{Corollary}

\usepackage{array}
\usepackage{collcell} %
\newcommand{\mycell}[1]{\parbox[c]{\linewidth}{\centering #1}}
\newcolumntype{C}[1]{>{\collectcell\mycell}p{#1}<{\endcollectcell}}

\newcommand{\appref}[1]{%
  \ifprlmode
    Supplemental Material%
  \else
    \cref{#1}%
  \fi
}

\begin{document}

\title{Exponentially Improved Constant in Quantum Solution Extraction}
\author{Gumaro Rendon}
\email[]{grendon@fujitsu.com}
\affiliation{Fujitsu Research of America, Inc., 4655 Great America Pkwy, Santa Clara, CA 95054, USA}

\begin{abstract}
We have provided an algorithm to extract a smooth and positive definite function $\psi(x)$ encoded in quantum memory of size $2^n$ without running into the problem of exponentially suppressed sub-normalization. Through this, we remove an important bottleneck of solution information extraction, the last step, in fully solving an important class of differential equations on quantum computers. This class of problems includes solutions to the heat equation or other diffusive equations in fluid dynamics and finance.
\end{abstract}

\maketitle

\section{Introduction}

There is a long standing problem, an often neglected problem, in quantum computing algorithms for linear systems and differential equations: solution extraction. It all goes back to that famous work~\cite{harrow2009quantum}, in which authors achieve exponential improvements with respect to problem size for the matrix inversion of a linear system. The assumption made in regards to the solution vector was that we did not need to know the solution itself, but features of it that can be extracted efficiently on a quantum computer. They argue that accessing each element to know the full solution spoils the exponential advantage.

Similar assumptions have been made since this seminal work for quantum linear systems or differential equations (e.g. \cite{Berry_2014,childs2017quantum,berry2017quantum,Childs2021highprecision}). This is done because explicit solution extraction from quantum memory may spoil the exponential advantage with respect to system size. We propose here an antidote for this problem. But first, we will briefly explain with a 1D abstraction why this problem exists.
Suppose that we have an unknown smooth function $\psi(x)$ that is normalized:
\begin{align}
    \int^{1}_{-1} \left|\psi(x)\right|^2 {\rm d} x = 1.
\end{align}
We assume that this function can be stored in a quantum memory with $n$ qubits the following way:
\begin{align}\label{eq:Upsi}
    \langle 0 |_a U_{\psi}|0 \rangle | 0 \rangle_a = a_{\psi} | \psi \rangle,
\end{align}
where
\begin{align}
    |\psi \rangle = \frac{1}{\sqrt{\mathcal{N}}} \sum^{2^n-1}_{j=0} \psi(x_j) |j \rangle,
\end{align}
$\mathcal{N} = \sum_x |\psi(x)|^2   \approx 2^{n-1}$, and $x_j = (2 j/2^n - 1) $.  Here, the ancillary register $\ket{0}_a$ is part of the algorithm to prepare the solution on the main register, and $a_{\psi}$ is the sub-normalization expected if the algorithm to prepare the target state is non-unitary. 

This is in many cases how we expect the final vector solution to be encoded in algorithms for quantum linear system problems (QLSP, See \cite{harrow2009quantum,Berry_2014,gilyen2019quantum,childs2020quantum}) or other means of solving differential equations on a quantum computer~\cite{fang2023time,berry2017quantum}. However, these algorithms usually stop short of extracting the solution information or assume that the required quantities, e.g. $\langle \psi | O | \psi \rangle$, are easy to extract. Here, we address the problem that is the exponential sub-normalization that happens when trying to extract the solution from a single site:
\begin{align}
    \langle \psi | j  \rangle \langle j | \psi \rangle \propto |\psi(x_j)|^2/ 2^n.
\end{align}
There is a modest quadratic improvement by using Quantum Amplitude Estimation (QAE) where the cost would go like $O\left( |\psi(x_j)|/\sqrt{2^n}\right)$. Moreover, the existing methods of tomography like classical shadows~\cite{huang2020predicting} do not help in this case because these would estimate each $\langle \psi | j  \rangle \langle j | \psi \rangle$ with additive error $\epsilon$ at a cost that goes like $O(1/\epsilon^2)$, but the sub-normalization $1/2^n$ remains, this would mean that we would need to adjust the target $\epsilon$ accordingly to compensate for this exponential sub-normalization. 

In this work, we aim to improve this exponentially and not just quadratically. For this, we also assume that $\psi(x)>0$ and that $\psi(x)$ is analytic on the interval of interest with derivatives bounded by $\left| \frac{{\rm d}^k\psi(x)}{{\rm d} x^k} \right| \leq (\Lambda )^{k+1}$
where $\Lambda_k =  \max_{x,j\in [k]} \left(\psi^{(j)}(x)\right)^{1/(j+1)}$ where $[k]:= \{j \in \mathbb{N} | j \leq k\}$, and $\Lambda = \Lambda_{\infty}$. 

The main findings of our work here are summarized in the following theorem:
\begin{theorem}\label{thm:main}
Provided an unknown analytic function function $\psi(x)$ for $x \in [-1,1] $ that is normalized:
\begin{align*}
    \int^{1}_{-1} \left|\psi(x)\right|^2 {\rm d} x = 1,
\end{align*}
and whose derivatives are bounded by $\left| \frac{{\rm d}^j\psi(x)}{d x^j} \right| \leq \Lambda^{j+1} $, which can be stored in a quantum memory with $n$ qubits the following way:
\begin{align*}
    \langle 0 |_a U_{\psi}|0 \rangle | 0 \rangle_a = a_{\psi} | \psi \rangle,
\end{align*}
where
\begin{align*}
    |\psi \rangle &= \frac{1}{\sqrt{\mathcal{N}}} \sum^{2^n-1}_{j=0} \psi(x_j) |j \rangle,
\end{align*}
 $\mathcal{N} = \sum_j |\psi(x_j)|^2 $, and $x_j = 2 j/2^n -1$, one can estimate it at a quantum gate cost that goes like
 \begin{align*}
    \tilde{O}\left(\frac{1}{a_{\psi}}\frac{\Lambda^2 n^2 }{\epsilon_{\rm total}} \frac{\max_{x} \psi(x) }{\min_x \psi(x)} \right).
\end{align*}
    
\end{theorem}

\section{Extracting the amplitudes from the quantum memory \label{sec:mock_cheb}}

The method can be split into four important steps: 
\begin{itemize}
    \item Sample the squared integral of the solution.
    \item Interpolate those points.
    \item Differentiate that interpolant to obtain the integrand.
    \item And finally, take the square root of that approximation to obtain back the solution.
\end{itemize}
In what follows, we will describe what each of those steps consists of.

\subsection{Squared amplitude integral estimation with binary segmentation \label{subsec:integral_estimation}}

Here, we summarize the results from squared amplitude estimation with which we avoid the scaling factor $1/2^n$ because it gets absorbed into the very definition of the integral (the differential factor). 

\begin{lemma}[Probability Integral Estimation]\label{lem:PIE}
Given an analytic function $\psi(x)$, and statevector with positive-definite amplitudes $|\psi \rangle = \frac{1}{\sqrt{\mathcal{N}}}\sum^{2^n-1}_{j=0}  \psi(x_j) |j\rangle$, and a unitary such that $\tilde{U}_{\psi}|0\rangle=|\psi\rangle$, where $x_j=2j/2^n-1$, we can estimate the integral
\begin{align*}
    \Psi(\hat{x}) = \int^{\hat{x}}_{-1} \psi^2(x') {\rm d} x'
\end{align*}
where $\hat{x}=2 X/2^n -1 $ , and $X\in \{0,1,2,\dots,2^n-1\}$, within an error $\epsilon$ with a gate complexity of
\begin{align*}
    \tilde{O}\left(\frac{n^2 \max_x \psi(x) }{\epsilon_{\Psi}}\right),
\end{align*}
provided that $\epsilon_{\Psi} \leq n \max_x \psi^2 (x)$.

\end{lemma}

For a full proof see \appref{app:PIE}. If instead of the unitary $\tilde{U}_{\psi}$, we have $U_{\psi}$ from \cref{eq:Upsi} this follows:

\begin{corollary}[Non-Unitary Probability Integral Estimation]\label{cor:PIE_subnorm}
Given an analytic function $\psi(x)$, and statevector with positive-definite amplitudes $|\psi \rangle = \frac{1}{\sqrt{\mathcal{N}}}\sum^{2^n-1}_{j=0}  \psi(x_j) |j\rangle$, and a unitary such that $\langle 0 |_a U_{\psi} | 0 \rangle_a|0\rangle  = a_{\psi}|\psi\rangle$, where $x_j=2j/2^n-1$, we can estimate the integral $\Psi(\hat{x}) = \int^{\hat{x}}_{-1} \psi^2(x') {\rm d} x'$ where $\hat{x}=2 X/2^n -1 $ , and $X\in \{0,1,2,\dots,2^n-1\}$, within an error $\epsilon$ with a gate complexity of
\begin{align*}
    \tilde{O}\left(\frac{1}{a_{\psi}}\frac{n^2 \max_x \psi(x) }{\epsilon_{\Psi}}\right),
\end{align*}
provided that $\epsilon_{\Psi} \leq n \max_x \psi^2 (x)$.

\end{corollary}

With these results, we are certain of being able to estimate the integrals $\Psi(x)$ without having to run QAE for  $\sim 2^n$ points separately. However, if we want to have the functional form of $\Psi(x)$ throughout the whole interval we must perform some kind of interpolation.

\subsection{Estimating the integral Cebyshev expansion with a few evaluations\label{subsec:cheb_interp}}

Here, we explain how we obtain the integral of the squared amplitude, $\Psi(x)$, for the whole interval with just a few samples. Recalling that
\begin{align}\label{eq:Phi}
\Psi(x)=\int^{x}_{-1} \psi(x')^2 {\rm d} x',
\end{align}
we can now use Chebyshev interpolation to obtain the expansion coefficients of the interpolant through the following linear system: $\mathbf{V}  a = f$, where $(f)_k = \Psi(x_k)$, and $x_k$ are a choice of  interpolation nodes. We have that
\begin{align}
 &\mathbf{V} := \cr 
&\begin{pmatrix}
 u_0(x_{{\rm cheb},1})   & u_1(x_{{\rm cheb},1})   & \dots  & u_{M-1}(x_{{\rm cheb},1}) \\
 u_0(x_{{\rm cheb},2})   & u_1(x_{{\rm cheb},2})   & \dots  & u_{M-1}(x_{{\rm cheb},2}) \\
 \vdots     & \vdots     & \ddots & \vdots  \\
 u_0(x_{{\rm cheb},M}) & u_1(x_{{\rm cheb},M}) & \dots  & u_{M-1}(x_{{\rm cheb},M})
\end{pmatrix}. \cr 
\end{align}
With this, one can obtain the coefficients, $a_j$, with $a=\mathbf{V}^{-1}f$. Like in \cite{Rendon2024improvedaccuracy,rendon2024needtrotter}, the choice of interpolating nodes $x_{{\rm cheb},k}$ and the interpolating set of polynomials are the Chebyshev nodes and polynomials.
In that case, $u_j$ is defined by
\begin{align} \label{eq:cheb_orthonorm_def}
     u_j(x) :=
    \begin{cases}
    \sqrt{\frac1M}T_0(x), &j=0 \\
    \sqrt{\frac2M}T_j(x), &j=1,2,\dots,M-1 \\
    \end{cases} 
\end{align}
where $T_j$ is the standard $j$th Chebyshev polynomial, $T_j(x) := \cos (j \cos^{-1} x)$.
The node collocation is described by
\begin{align} \label{eq:cheb_node_def}
    x_{{\rm cheb},k} = \cos\left(\frac{2k-1}{2M} \pi\right), \quad k \in \{1,2,\dots,M\}.
\end{align}
These polynomials fulfill the discrete orthonormality condition \cite{optvan} with respect to the collocation nodes, that is,
\begin{align}
    \sum_{k=1}^M u_i (x_{{\rm cheb},k}) u_j (x_{{\rm cheb},k}) = \delta_{ij}
\end{align}
for all $0\leq i, j < M$.  With this we know that the condition number for $\mathbf{V}$ is optimal, or $\kappa\left(\mathbf{V}\right)= \sigma_{\rm max} \left(\mathbf{V}\right)/\sigma_{\rm min} \left(\mathbf{V}\right)=1$. The number of nodes needed for a target error $\epsilon_{\rm cheb}=\left| \Psi (x) - P_{M-1} \Psi(x) \right|$ is estimated and given in the following lemma:

\begin{lemma}[Chebyshev Integral Interpolation]\label{lem:M}
Given the following integral
\begin{align*}\label{eq:Phi}
\Psi(x)=\int^{x}_{-1} \psi(x')^2 {\rm d} x',
\end{align*}
where $\left| \frac{{\rm d}^k\psi(x)}{{\rm d} x^k} \right| \leq (\Lambda )^{k+1}$ for $x\in [0,1]$, we can interpolate $\Psi(x)$ through
\begin{align*}
   P_{M-1} \Psi(x) = \sum^{M-1}_{j=0} a_j u_j (x),
\end{align*}
with a target error $\epsilon_{\rm cheb}$, where 
\begin{align*}
     u_j(x) &:=
    \begin{cases}
    \sqrt{\frac1M}T_0(x), &j=0, \\
    \sqrt{\frac2M}T_j(x), &j=1,2,\dots,M-1, \\
    \end{cases} \cr
T_j(x) &:= \cos \left(j \cos^{-1} x\right),
\end{align*}
using
\begin{align*}
    M = O \left(\Lambda + \log (1/\epsilon_{\rm cheb} )\right)
\end{align*}
samples of $\Psi$, $(f)_k = \Psi(x_{{\rm cheb},k})$, with $x_{{\rm cheb},k} = \cos\left(\frac{2k-1}{2M} \pi\right)$.
Here, $a_j$ are obtained by solving the linear system:
\begin{align*}
    \mathbf{V}  a = f,
\end{align*}
where
\begin{align*}
 &\mathbf{V} := \cr 
&\begin{pmatrix}
 u_0(x_{{\rm cheb},1})   & u_1(x_{{\rm cheb},1})   & \dots  & u_{M-1}(x_{{\rm cheb},1}) \\
 u_0(x_{{\rm cheb},2})   & u_1(x_{{\rm cheb},2})   & \dots  & u_{M-1}(x_{{\rm cheb},2}) \\
 \vdots     & \vdots     & \ddots & \vdots  \\
 u_0(x_{{\rm cheb},M}) & u_1(x_{{\rm cheb},M}) & \dots  & u_{M-1}(x_{{\rm cheb},M})
\end{pmatrix}.\cr
\end{align*}
\end{lemma}

For a full derivation of these bounds see \appref{app:M}. With these results we have elucidated the path to obtain an interpolant of the integral $\Psi(x) = \int^{x}_{-1} \psi^2(x') {\rm d} x'$, now, to obtain the integrand one must just differentiate $\Psi(x)$ with respect to $x$. In what is next, we show the results for uncertainty propagation after differentiation and due to perturbed node collocation (we can't exactly sample the integral at Chebyshev nodes).

\subsection{Extracting the integrand $\psi^2$\label{subsec:integrand}}

Here, we explain how we obtain an approximation of the integrand $\psi^2$ out of the interpolant for $\Psi(x)$ (See \Cref{eq:Phi}). We also list the lemma that tells us how the errors, $\epsilon_{\Psi}$, from the estimated $\Psi(x)$ propagate towards this approximant of $\psi^2$ with error $\epsilon_{\psi^2}$ (for the full proof of \Cref{lem:derivative} see \appref{app:derivative}). 

First, we note that are actually going to obtain the coefficients from
\begin{align}
    \mathbf{V}_{\rm pert} a' = f',
\end{align}
where $f'$ is the integral samples with errors introduced from the finite-resources amplitude estimation and from using perturbed node collocations, $x_{{\rm mock},k}$. $\mathbf{V}_{\rm pert}$ is the corresponding slightly perturbed Chebyshev-Vandermonde matrix:
\begin{align}
&\mathbf{V}_{\rm pert} := \cr 
&\begin{pmatrix}
 u_0(x_{{\rm mock},1})   & u_1(x_{{\rm mock},1})   & \dots  & u_{M-1}(x_{{\rm mock},1}) \\
 u_0(x_{{\rm mock},2})   & u_1(x_{{\rm mock},2})   & \dots  & u_{M-1}(x_{{\rm mock},2}) \\
 \vdots     & \vdots     & \ddots & \vdots  \\
 u_0(x_{{\rm mock},M}) & u_1(x_{{\rm mock},M}) & \dots  & u_{M-1}(x_{{\rm mock},M})
\end{pmatrix}.\cr 
\end{align}
where the perturbation comes from obtaining the amplitudes at slightly-off collocation nodes, $x_{{\rm mock},k}$. This comes from using a finite size grid, even though exponential in $n$.

The error propagation from the integral estimation towards the integrand estimation can be summarized in the following lemma:
\begin{lemma}[Error Propagation to Derivative of Interpolant] \label{lem:derivative}
Provided we estimate $\Psi(x)$ through $\tilde{\Psi}(x)$ at the perturbed nodes $x_{{\rm mock},k}$ closest to the Chebyshev nodes
\begin{align*} 
    x_{{\rm cheb},k} = \cos\left(\frac{2k-1}{2M} \pi\right), \quad k \in \{1,2,\dots,M\},
\end{align*}
where $x_{{\rm mock},k} \in \{2(j/2^n)-1\,|\, j \in \mathbb{N} \,|\, j < 2^n\}$ with error $|\tilde{\Psi}(x_{{\rm mock},j}) - \Psi(x_{{\rm mock},j})|\leq \epsilon_{\Psi}$, and the truncation error from Chebyshev interpolation is $\left| \Psi(x) - P_{M-1} \Psi(x) \right|\leq \epsilon_{\rm cheb} \sim \epsilon_{\Psi}$, we can obtain an approximation of its derivative with
\begin{align*}
 \tilde{\psi}(x) = \sum^{M-1}_{j=0} j a'_j U_j(x) 
\end{align*}
within an error 
\begin{align}
O(M^2 \epsilon_{\Psi} ),
\end{align}
where $a'$ are the coefficients that solve the following perturbed Chebyshev-Vandermonde linear system:
\begin{align*}
\mathbf{V}_{\rm pert} a' = f',
\end{align*}
where $f_j = \tilde{\Psi} (x_{{\rm mock},j})$, and
\begin{align*}
&\mathbf{V}_{\rm pert} =\cr 
&\begin{pmatrix}
 u_0(x_{{\rm mock},1})   & u_1(x_{{\rm mock},1})   & \dots  & u_{M-1}(x_{{\rm mock},1}) \\
 u_0(x_{{\rm mock},2})   & u_1(x_{{\rm mock},2})   & \dots  & u_{M-1}(x_{{\rm mock},2}) \\
 \vdots     & \vdots     & \ddots & \vdots  \\
 u_0(x_{{\rm mock},M}) & u_1(x_{{\rm mock},M}) & \dots  & u_{M-1}(x_{{\rm mock},M})
\end{pmatrix}.
\end{align*}
and
\begin{align*} \label{eq:cheb_orthonorm_def}
     u_j(x) :=
    \begin{cases}
    \sqrt{\frac1M}T_0(x), &j=0 \\
    \sqrt{\frac2M}T_j(x), &j=1,2,\dots,M-1 \\
    \end{cases} .
\end{align*}
All of this, provided the perturbation on the nodes $|x_{{\rm cheb},j}- x_{{\rm mock},j}|= O(1/2^n)$ is small enough such that $\| \mathbf{V}_{\rm pert} ^{-1}\| = O(1)$, and $|\Psi(x_{{\rm cheb},j}) - \Psi(x_{{\rm mock},j})|= O(\Lambda^2/2^n) \sim \epsilon_{\Psi}$.
\end{lemma}

With these results, we have with us access to the integrand $\psi^2(x)$ and estimates on the resources required to achieve a target error on $\psi^2(x)$. That is, for some  target error $\epsilon_{\psi^2}$ we need to sample the integral $\Psi(x)$ with a tolerance $O(\epsilon_{\psi^2}/M^2)$ , where $M$ is expected to be $O(\log{1/\epsilon_{\rm cheb}})$. Finally, to approximate $\psi(x)$ we need to take the square root of the approximant of $\psi^2(x)$.

\subsection{Taking the square root of the approximant of $\psi^2$\label{subsec:square_root}}

Here, we briefly explain how we are able to take the square root of the approximant of $\psi^2$ without introducing any discontinuities and large uncertainties. Fr that, we will make use of the following lemma:

\begin{lemma}
Given a positive definite function $f(x)$ with finite first derivative for $x\in[-1,1]$ which we have estimated through $\tilde{f}(x)$ within error $\varepsilon$, we can estimate $\sqrt{f(x)}$ through
\begin{align}
    \sqrt{\tilde{f}(x)}
\end{align}
which would have an error with respect to $\sqrt{f(x)}$ that goes like
\begin{align*}
 \left| \sqrt{f(x)} - \sqrt{\tilde{f}(x)} \right|  =  O\left(\frac{\varepsilon}{\min_{s}\left(\sqrt{f(x)}\right)}\right)
\end{align*}
for some $\varepsilon \leq \min_{s}f(x)$.
\end{lemma}
\begin{proof}

We start by recalling the following bound for an analytic function $g(x)$:
\begin{align*}
       \left| g(x\pm \varepsilon) - g(x) \right| \leq  \varepsilon \max_{y\in [x-\varepsilon,x+\varepsilon]} \left( g'(y) \right).
\end{align*}
In order to use this identity, $\tilde{f}(x)$ must also be positive definite, thus we assume that $\varepsilon \leq \min_s f(x)$ such that the estimate $\tilde{f}(x)$ remains positive definite. Now, we look at the first derivative of $\sqrt{f(x)}$: $\frac{{\rm d} \sqrt{f(x)}}{{\rm d} f} = \frac{1}{2 \sqrt{f}}$.
With this, we know the error propagated to $\sqrt{\tilde{f}(x)}$ is 
\begin{align}
 \left| \sqrt{f(x)} - \sqrt{\tilde{f}(x)} \right|  =  O\left(\frac{\varepsilon}{\min_{s}\left(\sqrt{f(x)}\right)}\right)
\end{align}
for some $\varepsilon<\min_{s}f(x)$.
\end{proof}

Using this lemma, we find that the error propagated from taking the square root of the approximation of $\psi^2$, $2\sum_j j a'_j U_j(x)$, is $\epsilon_{\psi} = O\left(\frac{\epsilon_{\psi^2}}{\min_{x} \psi(x)} \right)$ provided that $\epsilon_{\psi^2} \leq \min_s \psi(x)$. Now, using the results from \cref{lem:derivative}, we know that at the same time the error from our estimates of $\Psi(x_{{\rm mock},k})$ propagate to $\psi(x)$ the following way: $\epsilon_{\psi} = O\left(\frac{\epsilon_{\psi^2}}{\min_{x} \psi(x)} \right) =\tilde{O}\left( \frac{M^2 \epsilon_{\Psi}}{ \min_{x} \psi (x)}\right)$.

Putting all of these results together with the quantum cost estimates from  \cref{lem:PIE,cor:PIE_subnorm}, we obtain that we can extract $\psi(x)$ from memory (a unitary $U_{\psi}$, more specifically) through the following cost $O\left(\frac{1}{a_{\psi}}\frac{M^2 n^2 }{\epsilon_{\rm total}} \frac{\max_{x} \psi(x) }{\min_x \psi(x)} \right)$.

Finally, we have that $M = O\left(\Lambda + \log(1/\epsilon_{\rm cheb})\right)$ from \cref{lem:M}. Now, if we choose $\epsilon_{\rm cheb} \sim \epsilon_{\rm total}$ we have: $\tilde{O}\left(\frac{1}{a_{\psi}}\frac{\Lambda^2 n^2 }{\epsilon_{\rm total}} \frac{\max_{x} \psi(x) }{\min_x \psi(x)} \right)$.
With this, we have proved \cref{thm:main}.

\section{Conclusion}

We have provided here a method to extract solutions from quantum memory with an exponential improvement over naive quantum amplitude estimation. The constraints are that $\psi(x)$ must be positive definite, analytic, and $\min_x\psi(x)$ is not very small. This type of solutions one straightforwardly encounters in diffusive differential equations like those encountered in option pricing and the heat equation.

In the future, we would like to extend this analysis to problems where there are zero crossings present, shock or other singularities and in more general setups.

\section{Acknowledgements}

 I would like to thank Dr. Sarvagya Upadhyay of Fujitsu Research of America, Inc., for his valuable feedback on the manuscript. I am also deeply grateful to Dr. Hirotaka Oshima and Dr. Yasuhiro Endo of Fujitsu Research of Japan, Inc., for their insightful comments and suggestions throughout the course of the project. Lastly, I would like to thank \v{S}t\v{e}p\'an \v{S}m\'\i d for his useful comments on the final draft. 

\ifprlmode

\else
  \appendix
\section{Perturbation on linear system matrix, $\mathbf{V}$}
\label{app:Vpert} 

Regarding the perturbed matrix $\mathbf{V}_{\rm pert}$, we assumed that the norm of its inverse was $\|\mathbf{V}_{\rm pert}^{-1}\|_2=O(1)$. The orginal Chebyshev-Vandermonde, $\mathbf{V}_{\rm cheb}$, matrix has all eigenvalues $\lambda_j=1$, hence the stability. However, we are limited to collocating the sampling nodes at equi-spaced distance. To get an idea on why we can assume $\|\mathbf{V}_{\rm pert}^{-1}\|_2=O(1)$, we first consider the upper bounds on eigen-value perturbation\cite{bauer1960norms}:
\begin{align}
    \left|\lambda_{j}-\lambda\right| \leq \| \mathbf{V}-\mathbf{V}_{\rm pert} \|_2.
\end{align}
Now, in order to estimate the size of the spectral perturbation on the matrix, we look at the perturbation on each of its elements:
\begin{align}
   \| \mathbf{V}-\mathbf{V}_{\rm pert} \|_{\rm max} &= \max_{k,j\leq M-1} \left| u_j(x_{{\rm cheb},k})- u_j(x_{{\rm mock},k})\right|\cr 
   &\leq |x_{{\rm cheb},k}-x_{{\rm mock},k}| \max_{x\in[-1,1]} |D_x u_j(x)|.
\end{align}Thus, the error on the collocation points is:
\begin{align}
    \left| x_{{\rm cheb},k} - x_{{\rm mock},k} \right| = O(1/N).
\end{align}
Moreover, we have already established that
\begin{align}
    \max_s\left| \frac{{\rm d} T_j (x)}{{\rm d} x} \right| = O(j^2).
\end{align}
This means the maximum norm on the perturbation is:
\begin{align}
    \|\mathbf{V}-\mathbf{V}_{\rm pert} \|_{\rm max} = O\left(\frac{M^{3/2}}{N}\right).
\end{align}
Thus, the corresponding spectral norm is:
\begin{align}
   \|\mathbf{V}-\mathbf{V}_{\rm pert} \|_{2}&\leq M \|\mathbf{V}-\mathbf{V}_{\rm pert} \|_{\rm max} \cr 
   &= O\left(\frac{M^{5/2}}{N}\right).
\end{align}
We expect this to be very small since both $N=2^n$  grow exponentially with the number of qubits. For that reason, it is safe to assume, for each variable, that the perturbation is small and thus $\|\mathbf{V}_{\rm pert}\|_2=O(1)$. These findings, summarized in a lemma read:

\begin{lemma}
    A Chebyshev-Vandermonde $M\times M$ matrix, constructed with $M$ nodes $s'_j$,  which are the closest points to the Chebyshev nodes $s_j$, and which also can be placed on an equispaced grid on $x\in [-1,1]$ with $N$ points, where $1-M^{5/2}/N= \Omega (1)$, has
\begin{align*}
\| V_{\rm pert}^{-1} \| = O(1).
\end{align*}
\end{lemma}

\section{Proof of Chebyshev Integral Interpolation Lemma}
\label{app:M}

\begin{proof}

    With this we can now use the convergence rates for Chebyshev interpolation. 
\begin{align}
    \epsilon_{\rm cheb} =\left| \Psi (x) - P_{M-1} \Psi(x) \right| \leq \frac{1}{2^{M-1} M! } \max_s  \left| \frac{{\rm d}^M\Psi(x)}{{\rm d} x^M} \right|
\end{align}
First, we note that
\begin{align}
    \frac{{\rm d} \Psi (x) }{{\rm d} x } = \psi^2.
\end{align}
Also, by Leibniz' general rule:
\begin{align}
    \frac{d^j \psi^2}{ d x^j} = \sum \binom{j}{k} \psi^{(j-k)}(x) \psi^{(k)} (x).
\end{align}
Thus, a bound is
\begin{align}
    \left| \frac{d^j \psi^2}{ d x^j} \right| \leq  (\Lambda)^{j+1} \sum \binom{j}{k} = (\Lambda)^{j+1} 2^j =  (2\Lambda)^{j+1}/2
\end{align}
Thus,
\begin{align}
    \epsilon_{\rm cheb} =\left| \Psi (x) - P_{M-1} \Psi(x) \right| &\leq \frac{1}{2^{M-1} M! } \max_s  \left| \frac{{\rm d}^{M-1}\psi^2(x)}{{\rm d} x^{M-1} } \right| \cr 
     &\leq \frac{1}{ M! }   \left(\Lambda\right)^M.
\end{align}
Using Stirling bounds for $M!$, for $M\geq 1$, we find that
\begin{align}
    \epsilon_{\rm cheb} =\left| \Psi (x) - P_{M-1} \Psi(x) \right| &\leq \frac{1}{  \sqrt{2\pi M} (M/e)^M }   \left(\Lambda \right)^M \cr
    &\leq \frac{1}{   (M/e)^M }   \left(\Lambda \right)^M \cr
    &=  \exp\left( - M \left(\log M   - 1 - \log \Lambda \right)\right)
\end{align}
We can write this inequality the following way:
\begin{align}
   \frac{ M}{\Lambda e} \log \left( \frac{M}{\Lambda e}\right) \leq \frac{ \log \frac{1}{\epsilon_{\rm cheb}} }{\Lambda e} 
\end{align}
We note that this has the following solution:
\begin{align}
    M \leq {\Lambda e}\exp\left(W_0\left(\frac{\log(1/\epsilon_{\rm cheb})}{\Lambda e}\right)\right)
\end{align}
where $W_0$ stands for the Lambert-W function. We use the following bound
\begin{align}
    W_0(x) &\leq \log \left(\frac{2x+1}{1+\log(x+1)}\right)
\end{align}
For $x>0$ we can use simpler looser bound:
\begin{align}
    W_0(x) &\leq \log \left(2x+1\right).
\end{align}
Using this last bound we obtain:
\begin{align}
    M &\leq {\Lambda e}\exp\left(\log\left(\frac{2\log(1/\epsilon_{\rm cheb})}{\Lambda e}+1\right) \right)\cr
    &=  {\Lambda e} + 2\log(1/\epsilon_{\rm cheb})
\end{align}
So the number of nodes $M$ corresponding to a target error $\epsilon_{\rm cheb}$ is bounded through the inequality above. That means, that to ensure we have an error lower than a target error of $\epsilon_{\rm cheb}$ we must choose an $M \geq {\Lambda e} + \log(1/\epsilon_{\rm cheb})$. In any case, we must choose
\begin{align}
    M = O \left(\Lambda + \log (1/\epsilon_{\rm cheb} )\right).
\end{align}

\end{proof}

\section{Proof of Lemma on the Probability Integral Estimation}
\label{app:PIE}

\begin{proof}
We can estimate the integral with the following left-Riemann sum:
\begin{align}
    \Psi(\hat{x}) =  \frac{1}{2^{n-1}}\sum^{X-1}_{j=0} \psi^2(x_j) + O\left(\frac{\Lambda^2}{2^n}\right).
\end{align}
We can further estimate this sum with the re-normalized one
\begin{align}
    \Psi(\hat{x}) =  \frac{1}{\mathcal{N}}\sum^{X-1}_{j=0} \psi^2(x_j) + O\left(\frac{\Lambda^2}{2^n}\right),
\end{align}
because $\mathcal{N}/2^{n-1}=1+O(\Lambda^2/2^n)$. This last statement can be corroborated by thinking of $\mathcal{N}/2^n$ as the left-Riemann sum approximation of $\int^1_0 \psi^2(x') {\rm d} x'=1$ through:
\begin{align*}
    \mathcal{N}/2^{n-1} = \frac{1}{2^{n-1} }\sum^{2^n-1}_{j=0} \psi^2(x_j) = 1 + O(\Lambda^2/2^n).
\end{align*}
Here, we are consider $\hat{x}= 2 X / 2^n - 1$, for which we can always find the following binary decomposition:
\begin{align}
    X = X_0 + 2X_1 + \dots + 2^{m-1}X_{m-1},
\end{align}
where $X_j\in \{ 0,1\}$ and we truncate this expansion so that $X_{m-1}=1$.  Thus, the procedure for estimating this sum of probabilities is:

\begin{algorithm}[H]
\caption{Shift-Based Probability Estimation}
\SetAlgoLined
\DontPrintSemicolon
\KwInput{Binary decomposition $X = \sum_j X_j 2^j$, samples of $\psi(x)$ stored in quantum amplitudes}
\KwOutput{Estimate of $\displaystyle \frac{1}{\mathcal{N}} \sum_{j=0}^{X-1} \psi^2\left(x_j\right)$}

$W \gets 0$\tcp*{Running index to track cumulative shift}

\For{$p \in \{m{-}1, m{-}2, \dots, 0\}$}{
  \If{$X_p = 1$}{
    Estimate the quantity:
    \[
    q_p =\sqrt{\frac{1}{\mathcal{N}} \sum_{k=W}^{W+2^p - 1} \psi^2\left( x_k \right)}
    \]
    with precision $\varepsilon_p$ by measuring the probability (or its square root for QAE) of all zeros from bit $p$ to the right,\;
    after shifting the register by $w$ using $S_{-W}$:\;
    \[
    S_{-W} |\psi\rangle = \frac{1}{\sqrt{\mathcal{N}}} \sum_{k=0} \psi(x_k + W/2^n)\,|k\rangle
    \]
    $w \gets W + 2^p$\;
  }
}
\end{algorithm}

Thus, we only need to do at most $m \leq n$ QAE routines with target error $\epsilon$ to obtain an estimate of 
\begin{align}
   \sum^{m-1}_{p=0} q_p^2 = \frac{1}{\mathcal{N}}\sum^{X-1}_{j=0} \psi^2(x_j).
\end{align}
The error on each estimate $q^2_p$ is bounded by:
\begin{align}
    \varepsilon_p^2 + 2 q_p \varepsilon_p.
\end{align}
Moreover, we have the following bound:
\begin{align}
    q_p^2 \leq \sum_p q_p^2 \sim \Psi(\hat{x}) \leq \left(1+|\hat{x}|\right) \max_{x\leq \hat{x}} \psi^2(x).
\end{align}
With this, and the fact that $(1+|\hat{x}|)\leq 2$, we know that $q_p \lesssim  \sqrt{2}\max_{x\leq \hat{x}} \psi(x) $, which means the error on each estimate of $q_p^2$ is
\begin{align}
    O\left(\varepsilon_p^2 + \varepsilon_p  \max_{x} \psi(x)\right).
\end{align}
From now on, we assume that all $\varepsilon_p=\varepsilon$. By repeated triangle inequality, we have that the total error on our estimate of $\sum_p q_p^2$ is
\begin{align}
    O\left(m \left(\varepsilon^2 + \varepsilon  \max_{x} \psi(x)\right)\right).
\end{align}
Using the fact that $m\leq n$ we bound the error through the more convenient expression:
\begin{align}
    O\left(n \left(\varepsilon^2 + \varepsilon  \max_{x} \psi(x)\right)\right).
\end{align}
Finally, assuming that $\varepsilon < \max_{x} \psi(x)$ we have:
\begin{align}
    O\left(n \varepsilon \max_x \psi(x) \right).
\end{align}
Thus, if we want a target error $\epsilon_{\Psi}$ on the total integral, using QAE we need to target an error $\varepsilon_p = \varepsilon = \epsilon_{\Psi}/ (n \max_x \psi(x) )$ for each of the $q_p$. The cost of each QAE for every $p$ run scales like $O(n \max_x \psi(x) /\epsilon_{\Psi})$. This gives a total gate cost that is going to scale like:
\begin{align}
    O\left(\frac{n^2 \max_x \psi(x) }{\epsilon_{\Psi}}\right),
\end{align}
provided $\varepsilon = \epsilon_{\Psi} / (n \max_x \psi(x) ) \leq \max_x \psi(x)$.

\end{proof}

\section{Proof of Lemma on Error Propagation to the Derivative of Interpolant}
\label{app:derivative}

First, we look at the error on the vector of amplitudes coming from amplitude estimation and perturbed collocation nodes
\begin{align}
    \left| (f- f')_j \right| =O(  \epsilon_{\Psi} + {\Lambda^2}|x_{{\rm cheb},j}- x_{{\rm mock},j}|) ,
\end{align}
where $\epsilon_{\Psi}$ is the target error with which we estimate each $\Psi(s'_j)$ and the term $\Lambda^2|x_{{\rm cheb},j}-x_{{\rm mock},j}|$ is coming from the perturbation of the nodes. We now assume that $\Lambda^2 |x_{{\rm cheb},j}-x_{{\rm mock},j}| = O(\epsilon_{\Psi})$, which is a loose assumption since $|x_{{\rm cheb},j}-x_{{\rm mock},j}|=O(1/2^n)$, an exponentially suppressed quantity. With this, we look at how the error is propagated to the expansion coefficients:
\begin{align}
    \| a - a' \|_1 = \| V_{\rm perturbed}^{-1} (f-f') \|_1
\end{align}
Using Schwarz inequality
\begin{align}
    \| a - a' \|_1 &= \| V_{\rm perturbed}^{-1} (f-f') \|_1\cr 
    & =\| V_{\rm perturbed}^{-1} \|_2 \| (f-f') \|_2 \cr
    &= O\left(\sqrt{M} \epsilon_{\Psi}\right)
\end{align}
Here, we have assumed $\| V_{\rm perturbed}^{-1}\|_2=O(1)$ (See \cref{app:Vpert} for proof).
Ignoring the algorithmic error from using a finite size $M$ interpolation, we know that the propagated uncertainty is upper bounded by:
\begin{align}
    \|(a-a') D_x u(x)\|_1.
\end{align}
Here, 
\begin{align}
(D_x u(x))_j = \frac{{\rm d}(u(x))_j}{{\rm d} x}=\begin{cases}
    \sqrt{\frac{\pi}{M}} \frac{{\rm d}T_j(x)}{{\rm d} x} \quad j=0 \\
    \sqrt{\frac{\pi}{M}}\frac{{\rm d}T_j(x)}{{\rm d} x} \quad j \in \{1,\dots,M-1\}
\end{cases}. 
\end{align}
Moreover, we note that the derivative of $T_j(x)$ is
\begin{align}
    \frac{{\rm d} T_j(x)}{ {\rm d} x} = j U_{j-1}(x),
\end{align}
where $U_j(x)$ are the Chebyshev polynomials of second kind:
\begin{align}
    U_j ( x) = \frac{\sin\left((n+1) \arccos(x) \right)}{\sin(\arccos(x))}.
\end{align}
We also note that
\begin{align}
   \max_x |U_j (x)| \leq j+1,
\end{align}
which in turn means that
\begin{align}
   \max_{x} \left| \frac{{\rm d} T_j (x)}{{\rm d} x} \right| \leq j^2.
\end{align}
Thus, after differentiating the interpolant, we expect to have an error
\begin{align}
   \epsilon_{\psi^2} =  O\left(M^{2}\epsilon_{\Psi}\right).
\end{align}

\fi

\end{document}